\documentclass[11pt]{scrartcl}  

\usepackage[utf8]{inputenc}
\usepackage[T1]{fontenc}
\usepackage[english]{babel}
\usepackage{lmodern}
\usepackage{hyperref} 
\usepackage{csquotes} 
\usepackage[onehalfspacing]{setspace}

\usepackage{algorithm}
\usepackage[noend]{algpseudocode}
\algrenewcommand\algorithmicrequire{\textbf{Input:}}
\algrenewcommand\algorithmicensure{\textbf{Output:}}
\algnewcommand{\algorithmicand}{\textbf{ and }}
\algnewcommand{\algorithmicor}{\textbf{ or }}
\algnewcommand{\OR}{\algorithmicor}
\algnewcommand{\AND}{\algorithmicand}

\usepackage{amsfonts,amssymb,amsmath,amsthm}
\usepackage{graphicx}  
\usepackage{marginnote}
\usepackage[dvipsnames]{xcolor}
\usepackage{authblk}
\usepackage{enumitem}
\usepackage{tabularx}
\usepackage{subfig}
\usepackage{mathtools}
\usepackage{bm}
\usepackage{cancel}
\usepackage{stmaryrd}

\theoremstyle{definition}
\newtheorem{theorem}{Theorem}[section]

\newtheorem{corollary}[theorem]{Corollary}
\newtheorem{definition}[theorem]{Definition}
\newtheorem{remark}[theorem]{Remark}
\newtheorem{lemma}[theorem]{Lemma}

\newtheorem{example}[theorem]{Example}

\newtheorem{assumption}[theorem]{Assumption}

\DeclareMathOperator*{\argmin}{arg\,min}

\newcommand{\precqq}{\begin{matrix}
		\hspace{1mm}\prec{ }\\[-3mm]
		\hspace{1mm}= { }
\end{matrix}}
\newcommand{\precnqq}{\not\hspace{-2mm}\precqq }

\usepackage{tikz, pgfplots}
\usetikzlibrary{arrows, shapes.misc, chains, patterns, shapes,backgrounds,fit}
\usetikzlibrary{decorations.pathreplacing}
\usepackage{xifthen}
\newcommand{\preceqq}{\!\begin{array}{c} \prec\\[-0.3cm]=\end{array}\!\!\vspace*{-0.3cm}}

\tikzset{
  candidat/.style={rectangle, inner sep=0pt, minimum size=0.1cm, draw=gray, fill=gray},
  nds/.style={circle, inner sep=0pt, minimum size=0.12cm, draw=black, fill=black},
  nodestyle/.style={circle, inner sep=0pt, minimum size=0.1cm, draw=black!30, fill=black!30},
  shelterstyle/.style={circle, inner sep=2pt, minimum size=0.22cm, draw=black, fill=gray!30, font=\tiny},
  evacstyle/.style={circle, inner sep=2pt, minimum size=0.22cm, draw=black, fill=gray!30, font=\tiny},
  ndns/.style={rectangle, inner sep=0pt, minimum size=0.12cm, draw=black, fill=black},
  test/.style={circle, inner sep=0pt, minimum size=0.12cm, draw=black, fill=black},
  c1/.style={black!30, line width=0.5mm},
  c2/.style={black!30, dashed, line width=0.4mm},
  c3/.style={black!30, dotted, line width=0.4mm},
}
\pgfplotsset{compat=1.8}

\usepackage{pgfplotstable, tabularx, booktabs} 
\usepackage{ltablex}
\begin{document}

\newcommand{\cmnt}[1]{\noindent {\bf[ #1 ]}}
\newcommand{\hide}[1]{}

\title{Shortest Paths with Ordinal Weights}
\subtitle{}

\author[(1)]{Luca E. Sch\"afer\footnote{Corresponding author\\Email address: luca.schaefer@mathematik.uni-kl.de (Luca Elias Schäfer)
}}
\author[(1)]{Tobias Dietz}
\author[(1)]{Nicolas Fr\"ohlich}
\author[(1)]{Stefan Ruzika}
\author[(2)]{Jos\'{e} Rui Figueira}

\affil[(1)]{Department of Mathematics, Technische Universit\"at Kaiserslautern, Germany}
\affil[(2)]{CEG-IST, Instituto Superior T\'{e}cnico, Universidade de Lisboa, Portugal}

\date{}
\maketitle

	\begin{abstract}
		\subsection*{Abstract}
			We investigate the single-source-single-destination "shortest" paths problem in acyclic graphs with ordinal weighted arc costs. We define the concepts of ordinal dominance and efficiency for paths and their associated ordinal levels, respectively. Further, we show that the number of ordinally non-dominated paths vectors from the source node to every other node in the graph is polynomially bounded and we propose a polynomial time labeling algorithm for solving the problem of finding the set of ordinally non-dominated paths vectors from source to sink.
			\newline
			
			\noindent\emph{Keywords: Networks, Ordinal scale, Ordinal shortest path problem, Multicriteria optimization, Non-dominance}
	\end{abstract}

\section{Introduction}
Shortest path problems and applications have been intensively studied in the literature, (cf. \cite{gallo}, \cite{cher}, \cite{drey}, \cite{sprev}). However, in practical applications, one may have qualitative or ordinal information instead of numerical data available, e.g., in the case of demonstration marches the police staff may be able to assess different paths or path segments as "secure", "insecure", or "neutral", respectively. Thus, instead of assessing the value of a path as the sum of the values of its corresponding arcs, we evaluate a path as a vector containing the ordinal arc weights associated with this path. Hence, we compare different paths by their corresponding vectors of ordinal arc weights.
In this work, we aim to investigate the case of having ordinal information on the set of arcs and how to find "good" paths with respect to a given ordinal scale. 

Literature on optimization problems with an ordinal scale is rather limited.
In \cite{happy}, the authors suggest a procedure for finding different routes which are "emotionally pleasant" based on data from a crowd-sourcing platform. This data is then translated into quantitative measures of location perceptions.
In \cite{spanjaard}, a general preference-based framework for combinatorial problems is presented, where an arbitrary order relation is assumed to be given. The authors showed that the problem of finding the set of preferred paths with respect to the given order relation is in general intractable. In their work, the authors study different characteristics of preference relations. For this purpose, they introduce the "independence axiom" to investigate on which class of preference structures their proposed algorithms yield efficient solutions. For combinatorial problems with a non-independent preference relation, they propose approximation algorithms, which yield supersets of the set of preferred paths. If weak independence is satisfied, a subset of such paths can be computed. In contrast to their work, we investigate general acyclic graphs with a fixed (non-independent) preference relation. The algorithm developed in this paper computes the entire set of ordinally non-dominated paths vectors in polynomial time.
In \cite{minpaths}, the authors analyze acyclic graphs whose arc set is partially ordered by a preference relation. Paths are evaluated componentwise while exploiting Bellman's principle of optimality, which is not fulfilled in our model.
The term ordinal efficiency is also used in the context of stochastic dominance, (cf. \cite{ordinaleff1}, \cite{ordinaleff2}, \cite{ordinaleff3}), which is not related to our work.
In \cite{bottle}, a generalization of combinatorial bottleneck problems using a partially ordered scale is studied, while in our model a specific preorder on the set of \(s\)-\(t\)-paths is used.
Motivated by the example of civil security, we introduce a model with a specific order relation which can be used for practical computations.

Thus, the remainder of this paper is outlined as follows. In Section \ref{sec:intro}, we introduce graph-theoretical concepts and basics of binary relations needed throughout this work. In Section \ref{sec:problem}, we present our concept of ordinal efficiency and dominance. A labeling algorithm with polynomial runtime is discussed in Section \ref{sec:algorithm}, while in Section \ref{sec:practicalImprovements} we propose practical improvements to reduce the running time in practical applications. In Section \ref{sec:computationalResults}, we show computational results of the algorithm on randomly created acyclic graphs and on acyclic grid graphs. In Section \ref{sec:conclusion}, we conclude the paper and propose some directions for further research.

\section{Preliminaries on Ordinal Weighted Graphs}\label{sec:intro}
Let $G = (V,E)$ denote a directed, connected and acyclic graph with vertex set $V$ and arc set $E \subseteq V \times V$. 
We assume the graph \(G\) to be without parallel arcs. Further, we suppose a set of qualitative complete ordered levels \(\mathcal{C} \coloneqq \{1,\ldots,K\}\) with \(1 \prec \cdots \prec K\) to be given, where \(x \prec y\) indicates that \(x\) is strictly preferred over \(y\).
We assign a qualitative level to every arc of \(G\) by introducing an ordinal weight function over the set of arcs, i.e., \(f \colon E \rightarrow \mathcal{C}\).
Furthermore, we assume a source node \(s\) and a sink node \(t\) to be given. A directed path \(P\)  from \(s\) to \(t\) is defined by a sequence of directed arcs connecting a series of nodes starting in \(s\) and ending in \(t\), i.e., \(P=(s,e_1,v_1,\ldots,v_{n-1},e_n,t)\). With \(\mathcal{P}\), we denote the set of all \(s\)-\(t\)-paths.

For the purpose of ordinal weighted graphs, we recall some definitions of binary relations, (cf. \cite{EHR}).
A binary relation on \(\mathcal{C}\) is a subset \(\mathcal{R}\) on \(\mathcal{C} \times \mathcal{C}\).
\begin{definition}[Properties of a binary relation]\mbox{}
	
	A binary relation \(\mathcal{R}\) on \(\mathcal{C}\) is called
	\begin{itemize}
		\item[1)] reflexive, if \((x,x) \in \mathcal{R}\) for all \(x \in \mathcal{C}\)
		\item[2)] transitive, if \((x,y) \in \mathcal{R}\) and \((y,z) \in \mathcal{R} \Rightarrow (x,z) \in \mathcal{R}\) for all \(x,y,z \in \mathcal{C}\)
		\item[3)] antisymmetric, if \((x,y) \in \mathcal{R}\) and \((y,x) \in \mathcal{R} \Rightarrow x = y\) for all \(x,y \in \mathcal{C}\)
	\end{itemize}
\end{definition}

\begin{definition}[Orders]\mbox{}
	
	A binary relation \(\mathcal{R}\) on \(\mathcal{C}\) is called a 
	\begin{itemize}
		\item[1)] preorder, if it is reflexive and transitive
		\item[2)] partial order, if it is reflexive, transitive and antisymmetric.
	\end{itemize}
	
\end{definition}
Given a preorder \(\preceq\) on \(\mathcal{C}\), we define two additional relations as follows:
\begin{align*}
	& x \prec y :\Leftrightarrow x \preceq y \text{ and } y \not\preceq x \text{ (asymmetric part of} \preceq)\\
	& x \sim y :\Leftrightarrow x \preceq y \text{ and } y \preceq x \text{ (symmetric part of} \preceq)
\end{align*}


\section{Problem formulation}\label{sec:problem}
In the sequel, we consider the single-source-single-destination "shortest" path problem, i.e., we aim to find the set of paths from a source node \(s\) to a target node \(t\) that "minimizes" the vectors of ordinal levels associated with the ordinal arc weights of these paths. Note, that these paths will be simple due to the acyclic nature of the graphs under consideration. A path is called simple if it contains no repeated vertices.
In order two distinguish paths with respect to their associated ordinal levels, we define the following concepts.
\begin{definition}[Ordinal path vector]\mbox{}
	
	Let \(P \in \mathcal{P}\) with \(P=(s,e_1,v_1,\ldots,v_{n-1},e_n,t)\). Then, the associated ordinal path vector is given by \(f(P) \coloneqq (f(e_1),\ldots,f(e_n)\).
\end{definition}

To be able to compare vectors of ordinal levels, we define the following orders on \(\mathcal{C}^n\), \(n \in \mathbb{N}\).
\begin{definition}[Comparison of two ordinal path vectors]\mbox{}
	
	Let \(x,y \in \mathcal{C}^n\) with \(n \in \mathbb{N}\). Then:
	\begin{align*}
	& x \prec y :\Leftrightarrow x^i \prec y^i \ \forall i = 1,\ldots,n\\
	& x \preceq y :\Leftrightarrow x^i \preceq y^i \ \forall i = 1,\ldots,n \text{ and } x^j \prec y^j \text{ for at least one }j\\
	& x \preceqq y :\Leftrightarrow x^i \preceq y^i \ \forall i = 1,\ldots,n
	\end{align*}
\end{definition}

In the following, we relate paths and their corresponding ordinal path vectors to each other. 

Further, we define the length of an ordinal path vector and its corresponding sorted version.
\begin{definition}[Notation for ordinal paths vectors]\label{def:notation}\mbox{}
	
	Let \(P \in \mathcal{P}\) with \(P=(s,e_1,v_1,\ldots,v_{n-1},e_n,t)\), let \(m < n\) with \(m,n \in \mathbb{N}\), and let \(\pi\) denote a permutation of the corresponding index set. Then:
	\begin{itemize}
		\item[1)] \(\text{len}(f(P))\) denotes the length of the ordinal path vector \(f(P)\)
		\item[2)] \(\text{sort}(f(P))\coloneqq(f(e_{\pi(1)}),\ldots,f(e_{\pi(n)}))\) with \(f(e_{\pi(1)})\preceq \dots \preceq f(e_{\pi(n)})\) denotes the vector \(f(P)\) with its associated ordinal levels ordered in a non-decreasing manner
		\item[3)] \(\text{sort}^{forw}_m(f(P)) \coloneqq(f(e_{\pi(1)}),\ldots,f(e_{\pi(m)}))\) with \(f(e_{\pi(1)}) \preceq \dots \preceq f(e_{\pi(m)})\) contains the first \(m\) entries of \(\text{sort}(f(P))\)
		\item[4)] \(\text{sort}^{backw}_m(f(P))\coloneqq(f(e_{\pi(n-m+1)}),\ldots,f(e_{\pi(n)}))\) with \(f(e_{\pi(n-m+1)}) \preceq \dots \preceq f(e_{\pi(n)})\) contains the last \(m\) entries of \(\text{sort}(f(P))\)
	\end{itemize}
\end{definition}

Note, that \(\text{sort}^{forw}_m(f(P))\) and \(\text{sort}^{backw}_m(f(P))\) denotes the \(m\) best and worst entries of \(f(P)\), respectively.

\begin{example}\mbox{}
	
	Let \(P \in \mathcal{P}\) with \(P = (s,e_1,v_1,e_2,v_2,e_3,t)\). Further, let \(f(e_1) = 2\), \(f(e_2) = 3\), and \(f(e_3) = 1\). Then, it holds:
	\begin{itemize}
		\item[1)] \(\text{len}(f(P)) = 3\)
		\item[2)] \(\text{sort}(f(P)) = (1,2,3)\)
		\item[3)] \(\text{sort}^{forw}_2(f(P)) = (1,2)\)
		\item[4)] \(\text{sort}^{backw}_2(f(P)) = (2,3)\)
	\end{itemize}
\end{example}

Using the above mentioned definitions, we are now able to state our concept of dominance and efficiency with respect to ordinally weighted paths.

\begin{remark}
	Note that if \(m=n\) in Definition \ref{def:notation} it holds that
	
	\(\text{sort}(f(P))=\text{sort}^{backw}_n(f(P))=\text{sort}^{forw}_m(f(P))\).
\end{remark}

\begin{definition}[Ordinal dominance]\label{def:dominance}\mbox{}
	
	Let \(P_1,\ P_2 \in \mathcal{P}\) and let \(m,\ n \in \mathbb{N}\) denote the length of \(f(P_1)\) and \(f(P_2)\), respectively. Then, $P_1$ ordinally dominates $P_2$, i.e., 
	\begin{equation*}
		P_1 \precqq P_2 :\Leftrightarrow 
		\begin{cases}
			\text{sort}(f(P_1)) \precqq \text{ sort}(f(P_2)), \text{ if } m = n\\
			\text{sort}^{backw}_n(f(P_1))  \precqq \text{ sort}(f(P_2)), \text{  if } m > n\\
			\text{sort}(f(P_1)) \precqq \text{ sort}^{forw}_m(f(P_2)), \text{ if } m < n
		\end{cases}
	\end{equation*}
\end{definition}

\begin{definition}[Ordinally efficient path]\mbox{}
	
	A path \(P \in \mathcal{P}\) is called ordinally efficient, if there is no other  \(P' \in \mathcal{P}\) such that \(P' \precqq P\) and \(\text{sort}(f(P')) \neq \text{ sort}(f(P))\). The set of ordinally efficient \(s\)-\(t\)-paths is denoted by \(\mathcal{P}_E\).
\end{definition}

\begin{definition}[Ordinally non-dominated path vector]\mbox{}
	
	A sorted ordinal path vector \(\text{sort}(f(P)),\ P\in \mathcal{P}\), is called ordinally non-dominated, if \(P\) is ordinally efficient. The set of ordinally non-dominated paths vectors from \(s\) to \(t\) is denoted by \(\mathcal{P}_N\).
\end{definition}

\begin{definition}[Minimal complete set]\mbox{}
	
	Let \(P_1, P_2 \in \mathcal{P}\). We say that \(P_1\) is equivalent to \(P_2\) if and only if \(\text{sort}(f(P_1)) = \text{sort}(f(P_2))\). A minimal complete set \(\mathcal{P}_E^* \subseteq \mathcal{P}_E\) is a set of ordinally efficient \(s\)-\(t\)-paths such that all \(P \in \mathcal{P}\backslash\mathcal{P}_E\) are either dominated or equivalent to at least one \(P \in \mathcal{P}_E\).
\end{definition}

The previous definitions are motivated by the following real world application in the context of civil security.

\begin{example}
	Consider an evacuation scenario where one aims to determine the best possible evacuation routes. Besides the length of those routes, also the quality or nature of the path segments have to be taken into consideration. Undoubtedly, broad asphalted streets are better than narrow stairways in case of an evacuation. This is the reason why a decision maker may be able to rate streets by ordinal criteria, e.g., "good", "moderate", and "straitened", where "good" \( \prec\) "moderate" \( \prec\)  "straitened". Moreover, in an evacuation scenario, the resulting network can be assumed to be acyclic, since people should always move from the incident away towards some assembly point. An example is given in Figure \ref{fig:example}. Arcs labeled by "good", "moderate", and "straitened" are drawn solid, dashed, and dotted, respectively. The two ordinally non-dominated paths from \(S\) to \(T\) are highlighted.

	\begin{figure}[h!]
		\centering
		\begin{tikzpicture}[scale=0.8, ->]
		
		\node[evacstyle] (0) at (3,3) {$S$};
		\node[nodestyle] (1) at (1.2,2.5) {};
		\node[nodestyle] (2) at (1,1) {};
		\node[nodestyle] (3) at (0,2) {};
		\node[nodestyle] (4) at (0.5,4.2) {};
		\node[nodestyle] (5) at (1.3,4.8) {};
		\node[nodestyle] (6) at (0.8,5.6) {};
		\node[nodestyle] (7) at (2.8,4.7) {};
		\node[nodestyle] (8) at (3.3,5.5) {};
		\node[nodestyle] (9) at (4.5,4.5) {};
		\node[nodestyle] (10) at (5.3,5.6) {};
		\node[nodestyle] (11) at (4.6,6.2) {};
		\node[nodestyle] (12) at (7,5.5) {};
		\node[nodestyle] (14) at (7.7,4.5) {};
		\node[nodestyle] (15) at (6.2,4.1) {};
		\node[nodestyle] (16) at (5,3.2) {};
		\node[nodestyle] (17) at (6.9,2.5) {};
		\node[nodestyle] (18) at (7.2,1.7) {};
		\node[nodestyle] (19) at (5.6,1.2) {};
		\node[nodestyle] (20) at (5.2,2.2) {};
		\node[nodestyle] (21) at (5.7,0.3) {};
		\node[nodestyle] (22) at (4.2,1.1) {};
		\node[nodestyle] (23) at (4.3,2.4) {};
		\node[nodestyle] (24) at (3.3,0.1) {};
		\node[nodestyle] (25) at (2.9,1.5) {};
		\node[shelterstyle] (26) at (8,3.4) {$T$};
		\path (0) edge[c2] (1);
		\path (0) edge[c1, black] (7);
		\path (0) edge[c1,black] (9);
		\path (0) edge[c2] (23);
		\path (0) edge[c1] (25);
		\path (1) edge[c2] (2);
		\path (1) edge[c1,black] (3);
		\path (1) edge[c2] (4);
		\path (2) edge[c1,black] (24);
		\path (3) edge[c1,black] (2);
		\path (3) edge[c1] (4);
		\path (4) edge[c1] (5);
		\path (4) edge[c3] (6);
		\path (5) edge[c2] (6);
		\path (5) edge[c2] (8);
		\path (6) edge[c3] (8);
		\path (6) edge[c2] (11);
		\path (7) edge[c2,black] (1);
		\path (7) edge[c3] (5);
		\path (7) edge[c3] (8);
		\path (8) edge[c1] (11);
		\path (8) edge[c2] (10);
		\path (9) edge[c2] (8);
		\path (9) edge[c3] (10);
		\path (9) edge[c2,black] (15);
		\path (9) edge[c2] (16);
		\path (10) edge[c2] (12);
		\path (10) edge[c1] (14);
		\path (11) edge[c1] (12);
		\path (12) edge[c3] (14);
		\path (14) edge[c3,black] (26);
		\path (15) edge[c1,black] (14);
		\path (16) edge[c3] (15);
		\path (16) edge[c3] (20);
		\path (16) edge[c3] (17);
		\path (17) edge[c2] (26);
		\path (18) edge[c2,black] (26);
		\path (18) edge[c1] (17);
		\path (19) edge[c1,black] (18);
		\path (20) edge[c2] (18);
		\path (20) edge[c1] (19);
		\path (21) edge[c2] (18);
		\path (21) edge[c1,black] (19);
		\path (22) edge[c2] (19);
		\path (22) edge[c1,black] (21);
		\path (23) edge[c2] (20);
		\path (23) edge[c3] (22);
		\path (23) edge[c1] (25);
		\path (24) edge[c3] (21);
		\path (24) edge[c2,black] (22);
		\path (25) edge[c2] (24);
		\path (25) edge[c3] (2);
		\end{tikzpicture}
		\caption{Application of "shortest paths with ordinal weights" in an evacuation scenario}
		\label{fig:example}
	\end{figure}
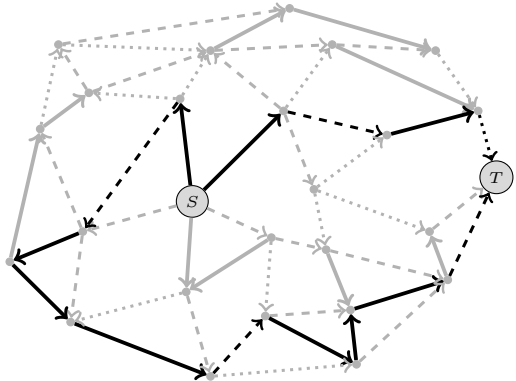
\end{example}

\begin{lemma}[Preorder]\mbox{}
	
	The relation \(\precqq\) defined on the set of \(s\)-\(t\)-paths \(\mathcal{P}\) is a preorder.
\end{lemma}
\begin{proof}
	Let \(P \in\mathcal{P}\) be an arbitrary \(s\)-\(t\)-path. Then, \(\text{sort}(f(P)) \precqq \text{ sort}(f(P))\) and, therefore, it is reflexive.
	To show that \(\precqq\) is transitive, let \(P, Q, R \in \mathcal{P}\) be \(s\)-\(t\)-paths with \(\text{len}(f(P)) = l\), \(\text{len}(f(Q)) = m\), \(\text{len}(f(R)) = n\) and \(P\precqq Q\) and \(Q\precqq R\). We have to show that \(P\precqq R\) and, therefore, we distinguish six cases:
	\begin{itemize}
		\item[1)] \(n \geq m \geq l\):
		We know that \(\text{sort}(f(P)) \precqq \text{ sort}^{forw}_l(f(Q))\) and
		\newline 
		\(\text{sort}(f(Q)) \precqq\text{ sort}^{forw}_m(f(R))\) due to \(P \precqq Q\) and \(Q \precqq R\).
		Therefore, it holds 
		\newline
		\((\text{sort}(f(Q)))_j \preceq (\text{sort}^{forw}_m(f(R)))_j\), \(j = 1,\ldots,m\) and consequently 
		\newline
		\(\text{sort}(^{forw}_l(f(Q)))_j \preceq (\text{sort}^{forw}_l(f(R))_j\), \(j = 1,\ldots,l\) due to \(l\leq m\).
		It follows that \(\text{sort}(f(P))= \text{sort}^{forw}_l(f(P)) \precqq \text{ sort}^{forw}_l(f(Q)) \precqq \text{ sort}^{forw}_l(f(R))\) and thus, \(P \precqq R\).
		
		\item[2)] \(l \geq m \geq n\):
		We know that \(\text{sort}^{backw}_m(f(P)) \precqq \text{ sort}(f(Q))\) and
		\newline
		\(\text{sort}^{backw}_n(f(Q)) \precqq \text{ sort}(f(R))\) due to \(P \precqq Q\) and \(Q \precqq R\).
		Therefore, it holds 
		\newline
		\((\text{sort}^{backw}_m(f(P)))_j \preceq (\text{sort}(f(Q)))_j\), \(j = 1,\ldots,m\) and consequently 
		\newline
		\((\text{sort}^{backw}_n(f(P)))_j \preceq (\text{sort}^{backw}_n(f(Q))_j\), \(j = 1,\ldots,n\) due to \(n<m\).
		It follows that \(\text{sort}^{backw}_n(f(P)) \precqq \text{ sort}^{backw}_n(f(Q)) \precqq \text{ sort}^{backw}_n(f(R)) = \text{ sort}(f(R))\) and thus, \(P \precqq R\).
		\item[3)] The remaining cases can be shown analogously.
	\end{itemize}
	This concludes the proof of transitivity and thus the relation \(\precqq\) defines a preorder on \(\mathcal{P}\).
	
\end{proof}

Note, that \(\precqq\) defined on the set of simple \(s\)-\(t\)-paths is not antisymmetric and thus it does not define a partial order.
To show that \(\precqq\) is not antisymmetric, consider the graph \(G=(V,E)\) with \(V=\{s,a,t\}\) and \(E=\{(s,a),(a,t),(s,t)\}\) as depicted in Figure~\ref{antisymmetric}.
Let \(P = (s,t)\) and \(Q = (s,a,t)\) and let \(j \in \mathcal{C}\) be an arbitrary ordinal level.  Note, that \(P \precqq Q\) and \(Q \precqq P\). However, \(\text{sort}(f(P)) \neq \text{sort}(f(Q))\).
\begin{figure}[h!]
	\centering
	\begin{tikzpicture}[scale=0.5]
	\node[draw, circle,inner sep=2pt] (s) at (0,10) {$s$};
	\node[draw, circle,inner sep=2pt] (a) at (5,8) {$a$};
	\node[draw, circle,inner sep=2pt] (t) at (10,10) {$t$};
	
	\draw[->,thick] (s)--node[above]{\(j\)}(t);
	\draw [->,thick] (s)--node[below]{\(j\)}(a);
	\draw [->,thick] (a)--node[below]{\(j\)}(t);
	\end{tikzpicture}
	\caption{Counterexample - antisymmetric}
	\label{antisymmetric}
\end{figure}
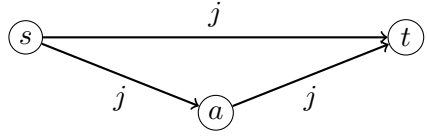

Throughout this paper we make the following assumptions.
\newline

\begin{assumption}\mbox{}
	\begin{itemize}
		\item[1)] The number of ordinal levels is fixed, i.e., \(K\) is fixed.
		\item[2)] We are only interested in the set of ordinally non-dominated path vectors from \(s\) to \(t\) and in a minimal complete set \(\mathcal{P}_E^*\).
	\end{itemize}
\end{assumption}

We state the problem of finding the set of \(s\)-\(t\)-paths that minimizes the ordinal path vectors associated with these \(s\)-\(t\)-paths, i.e., 
\begin{equation*}\tag{\(\mathsf{OSP}\)}
	\begin{array}{lll}
	\min_{\precqq} &\text{sort}(f(P))& \\
	\text{s.t.} & P \in \mathcal{P}
	\end{array}	
\end{equation*}
We call this problem the ordinal shortest paths problem  \((\mathsf{OSP})\).


\begin{theorem}[Number of ordinally efficient paths]\mbox{}
	
	The number of ordinally efficient paths from \(s\) to \(t\) may be exponential in the number of nodes.
\end{theorem}
\begin{proof}
We construct an instance, where \(|\mathcal{P}_E|\) is exponential in the number of nodes. Therefore, let \(G=(V,E)\) denote a directed, connected, and acyclic graph with \(V=\{s=v_1,\ldots,v_n=t\}\) and \(E=\{(v_i,v_{i+1}), i = 1,4,7,\ldots,n-3\}\cup\{(v_i,v_{i+2}), i = 1,4,7,\ldots,n-3\}\cup\{(v_i,v_{i+2}), i = 2,5,8,\ldots,n-2\}\cup\{(v_i,v_{i+1}), i = 3,6,9,\ldots,n-1\}\), where \(n-1\) is divisible by \(3\), i.e., \(n-1 \mod 3 = 0\). Further, let \(j \in \mathcal{C}\) denote an arbitrary ordinal level in \(\mathcal{C}\) and let \(f(e) = j\) for all \(e \in E\), see Figure~\ref{expeff}. There are \(\frac{4n-4}{3}\) arcs in the graph, if we construct the instance as described. 
	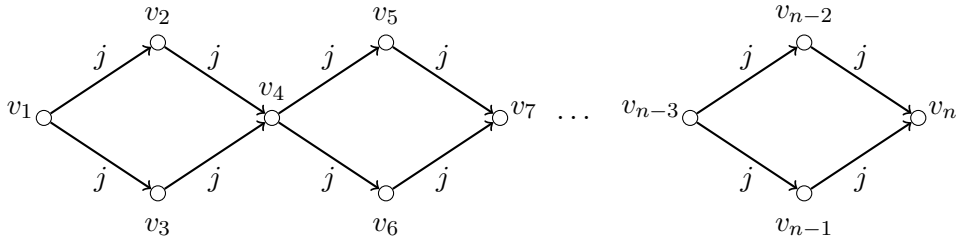
\begin{figure}[h!]
		\centering
		\begin{tikzpicture}[scale=0.5]
		\node[draw, circle,inner sep=2pt,label={[left]\(v_1\)}] (s) at (0,10) {};
		\node[draw, circle,inner sep=2pt,label={[above]\(v_2\)}] (a) at (3,12) {};
		\node[draw, circle,inner sep=2pt,label={[below,yshift=-0.3cm]\(v_3\)}] (b) at (3,8) {};
		\node[draw, circle,inner sep=2pt,label={[above]\(v_4\)}] (c) at (6,10) {};
		\node[draw, circle,inner sep=2pt,label={[above]\(v_5\)}] (d) at (9,12) {};
		\node[draw, circle,inner sep=2pt,label={[below,yshift=-0.3cm]\(v_6\)}] (e) at (9,8) {};
		\node[draw, circle,inner sep=2pt,label={[right]\(v_7\)}] (f) at (12,10) {};
		\node[] (dots) at (14,10) {$ \dots $};
		\node[draw, circle,inner sep=2pt,label={[left]\(v_{n-3}\)}] (g) at (17,10) {};
		\node[draw, circle,inner sep=2pt,label={[above]\(v_{n-2}\)}] (h) at (20,12) {};
		\node[draw, circle,inner sep=2pt,label={[below,yshift=-0.3cm]\(v_{n-1}\)}] (i) at (20,8) {};
		\node[draw, circle,inner sep=2pt,label={[right]\(v_n\)}] (t) at (23,10) {};
		
		\draw[->,thick] (s)--node[above]{\(j\)}(a);
		\draw [->,thick] (a)--node[above]{\(j\)}(c);
		\draw [->,thick] (s)--node[below]{\(j\)}(b);
		\draw [->,thick] (b)--node[below]{\(j\)}(c);
		\draw [->,thick] (c)--node[above]{\(j\)}(d);
		\draw [->,thick] (c)--node[below]{\(j\)}(e);
		\draw [->,thick] (d)--node[above]{\(j\)}(f);
		\draw [->,thick] (e)--node[below]{\(j\)}(f);
		\draw [->,thick] (g)--node[above]{\(j\)}(h);
		\draw [->,thick] (g)--node[below]{\(j\)}(i);	
		\draw [->,thick] (h)--node[above]{\(j\)}(t);
		\draw [->,thick] (i)--node[below]{\(j\)}(t);		
		\end{tikzpicture}
		\caption{An instance with exponential number of ordinally efficient paths.}
		\label{expeff}
	\end{figure}
Consequently, there are \(2^{\frac{n-1}{3}}\) paths from \(v_1\) to \(v_n\) with the same ordinal path vector \(f(P) = (j,\ldots,j)\) with len\((f(P))=\frac{2n-2}{3}\). Therefore, we obtain \(2^{\frac{n-1}{3}}\) ordinally efficient paths, which concludes the proof.

\end{proof}

However, we can show that the number of ordinally non-dominated paths vectors from \(s\) to every other node in the graph is polynomially bounded.
\begin{theorem}[Number of distinct ordinal paths vectors]\label{thm:sizelist}\mbox{}
	
	For all \(v \in V\), it holds that the number of distinct ordinal paths vectors from \(s\) to \(v\) is polynomially bounded by \(\mathcal{O}(n^K)\).
\end{theorem}
\begin{proof}
	The number of distinct ordinal path vectors corresponding to arc sets with exactly \(i\) arcs is the same as the number of possibilities to pick \(i\) elements from a set of \(K\) elements with replacement and without order. This is equal to \({K+i-1\choose i}\). Note, that this denotes an upper bound for the number of ordinal path vectors for node \(v\) of "length" \(i\). Further, a path contains at most \(n-1\) arcs. Summing up over all possible paths "lengths" results in
	\begin{align*}
	\sum_{i=1}^{n-1} {K+i-1\choose i} & = \frac{n{K+n-1\choose n}}{K} - 1 \\
	& = \frac{n(K+n-1)!}{K(K-1)!n!} -1 \\
	& = \frac{(n+K-1)(n+K-2)\cdot\ldots\cdot n}{K(K-1)!}-1 \in \mathcal{O}(n^K)
	\end{align*}
	The last term is polynomially bounded by \(\mathcal{O}(n^K)\) and so are the label sets at each node \(v \in V\).
	
\end{proof}

\begin{definition}[Absolute frequency]\mbox{}
	
	Let \(P \in \mathcal{P}\). Then, \(h(P,c)\) denotes the absolute frequency of \(c \in \mathcal{C}\) in \(P\), i.e., \(h(P,c) = |\{e \in P: f(e) = c\}|\). \(S(P)\) denotes the vector of the absolute frequencies of all ordinal levels with respect to \(P\), i.e., \(S(P) = (h(P,1),\ldots,h(P,K))^\top\).
\end{definition}

We define the following lexicographic order:
\begin{definition}[Lexicographic order]\mbox{}
	
	Let \(P_1,\ P_2 \in \mathcal{P}\). Then, \(S(P_1) \geq_{\text{lex}} S(P_2)\) if and only if \(h(P_1,i^*) > h(P_2,i^*)\) or \(S(P_1) = S(P_2)\), where \(i^* = \min\{i : h(P_1,i) \neq h(P_2,i)\}\).
\end{definition}

Now, we define \((\mathsf{LexMax})\) to be the optimization problem of finding the \(s\)-\(t\)-path whose corresponding absolute frequency vector is lexicographic maximal.
\begin{equation*}\tag{\(\mathsf{LexMax}\)}
\begin{array}{lll}
\max_{\text{lex}} &S(P)& \\
\text{s.t.} & P \in \mathcal{P}
\end{array}	
\end{equation*}

\begin{theorem}[Optimality]\mbox{}
	
	Let \(P^*\) be optimal for \((\mathsf{LexMax})\). Then \(\text{sort}(f(P^*))\) is ordinally non-dominated.
\end{theorem}
\begin{proof}
	Let \(P^*\) be optimal for \((\mathsf{LexMax})\) and assume there exists a path \(P \in \mathcal{P}\) such that \(P \precqq P^*\). Now, we have to distinguish three cases:
	\begin{itemize}
		\item[1)] Assume that \(|P| = |P^*|\), i.e., \(\text{sort}(f(P)) \precqq \text{ sort}f(P^*))\). Without loss of generality we can assume that \(\text{sort}(f(P)) \prec \text{ sort}(f(P^*))\), otherwise we are finished. Then, it exists \(i^* = \argmin_{i \in \{1, ..., |P| \}} \text{sort}(f(P))^i \neq \text{ sort}(f(P^*))^i\), where \(\text{sort}(f(.))^i\) denotes the \(i\)-th component of \(\text{sort}(f(.))\). Since the entries are sorted, it follows \(h(P, \text{ sort}(f(P))^{i^*})) > h(P^*, \text{ sort}(f(P))^{i^*}) \) and thus, \(S(P) \geq_{\text{lex}} S(P^*)\) which is a contradiction.
		\item[2)] The cases \(|P| > |P^*|\) and \(|P| < |P^*|\) are shown analogously.
	\end{itemize}
\end{proof}

Due to the last theorem, a path \(P\) maximizing \((\mathsf{LexMax})\) is ordinally efficient. We can find such a path by manipulating the arc weights of the graph and then computing a longest path. Note that the longest path problem can be solved in directed acyclic graphs in linear time (cf. \cite{schrijver}). For \(k \in \mathcal{C}\), let \(m_k\) denote the number of arcs \(e \in E \) such that \(f(e) = k\). Further, we denote by \(d_k\) the number of digits in order to store \(m_k\), which is given by \(d_k \coloneqq \lceil \log_{10}(m_k) \rceil\). We then can calculate the arc weight of arc \(e \in E\) with  \(f(e) = k\) in the following manner: \( w_e=  1 \cdot 10^{-1-\sum_{i=1}^{k-1} d_k} \). 
By construction, computing a longest path with respect to the new arc weights yields an optimal solution of \((\mathsf{LexMax})\). Note that the size of the modified instance is polynomial in the original input size.

\section{Labeling Algorithm}\label{sec:algorithm}
In this section, we propose a labeling algorithm with a label selection strategy to compute the set of ordinally non-dominated paths vectors from \(s\) to \(t\). During the execution of the algorithm, we store several labels at each node \(v \in V\) corresponding to different \(s\)-\(v\)-paths. Therefore, let \(v \in V\) be a node in \(G\) and let \(\mathcal{L}(v)\) denote the set of labels at node \(v\), where each label \(L \in \mathcal{L}(v)\) is a tuple \((f(P_L),\text{Pred}_L)\). The ordinal path vector associated with the \(s\)-\(v\)-path \(P_L\) with respect to label \(L\) is depicted by \(f(P_L)\) and \(\text{Pred}_L\) denotes a sorted list of nodes on the current \(s\)-\(v\)-path \(P_L\) with respect to label \(L\) at node \(v\). By \(\mathcal{L}\), we denote the set of all labels at all nodes, i.e., \(\mathcal{L} \coloneqq \bigcup_{v \in V} \mathcal{L}(v)\).
Moreover, we store a set of temporary labels, called \texttt{Temp}.
For the purpose of extending a label at a node \(v\) by an arc \((v,w)\), we define the following operation.
\begin{definition}[Label extension]\mbox{}
	
	Let \(L=(f(P_L),\text{Pred}_L)^\top\) be a label in \(\mathcal{L}(v)\) with \(v \in V\) and let \((v,w) \in E\). Then, \(L \oplus f(v,w) \coloneqq \left(\begin{array}{c} \left(f(P_L), f(v,w)\right)\\ \text{Pred}_L.\text{append}(w) \end{array}\right)\), where \(\text{Pred}_L.\text{append}(w)\) indicates that node \(w\) is appended at the end of the sorted list \(\text{Pred}_L\).
\end{definition}

Note that we are interested in a minimal complete set of ordinally non-dominated paths vectors from \(s\) to \(t\). Thus, we do not include a label at a specific node if its corresponding sorted ordinal path vector coincides with another sorted ordinal path vector already present at that node. Consequently, if we say that a label \(L\) is already contained in \(\mathcal{L}(v)\) for some \(v \in V\), we mean that \(\text{sort}(f(P_L))\) is equal to \(\text{sort}(f(P_{L'}))\) for some \(L' \in \mathcal{L}(v)\) although \(\text{Pred}_L\) and \(\text{Pred}_{L'}\) might differ.

The algorithm then works as follows.
Initially, we create a label at source node \(s\) with an empty ordinal path vector and a list containing only \(s\). We insert this label in the set of temporary labels, called \texttt{Temp}. In each iteration, we choose an arbitrary label \(L\) at node \(v\) from \texttt{Temp}, remove that label from \texttt{Temp}, investigate all outgoing arcs from node \(v\), and create a new label \(L'\) by extending the label \(L\) along the outgoing arc.
In the case that the end node of the outgoing arc is not the target node \(t\), we include the new label \(L'\) at the corresponding label set of the end node (for modifications, see Section \ref{sec:practicalImprovements}). This is due to the fact that Bellman's principle of optimality does not hold, cf. Example \ref{ex:alg+bellman}.
In the case that the end node of the outgoing arc is equal to \(t\), we temporarily set a binary flag of \(L'\) to \(1\). If this flag remains \(1\), we include \(L'\) in \(\mathcal{L}(t)\), otherwise not. This flag only remains to be equal to \(1\), if \(L'\) is not dominated by any other label in \(\mathcal{L}(t)\).
Note that in the case of the end node being equal or unequal to \(t\), we have to check if the new label is already present in the label set of the end node. This is true, because we aim to find the set of ordinally non-dominated paths vectors.
We iterate as long as the set of temporary labels is non-empty. Due to the fact that the number of ordinally non-dominated paths vectors from \(s\) to every other node in the graph  is polynomially bounded, we also know that the number of labels at each node is polynomially bounded throughout the algorithm. Further, in every iteration, we remove exactly one label from \texttt{Temp}. Thus, the algorithm will eventually terminate. After termination, the minimal complete set of ordinally efficient \(s\)-\(t\)-paths can be recovered by using the sorted list of nodes \(\text{Pred}_L\) for all \(L \in \mathcal{L}(t)\).

\begin{algorithm}
	\caption{Ordinal Labeling Algorithm}
	\textbf{Input:} A Digraph \(G=(V,E)\), loopless, no parallels, source \(s\), sink \(t\) and a function \(f: E \rightarrow \mathcal{C}\)\\
	\textbf{Output:} All ordinally non-dominated paths vectors from \(s\) to \(t\), i.e., \(\mathcal{P}_N\)\\
	\textbf{Initialization:} Create label \(L \coloneqq \left(\begin{array}{c} (\cdot)\\ \lbrack s\rbrack \end{array}\right)\) at node \(s\), i.e., \(\mathcal{L}(s) = \{L\}\), let\\\vspace{0.2cm} \mbox{\texttt{Temp} \(= \{L\}\)}
	\begin{algorithmic}[1]
	\While{\(\texttt{Temp} \neq \emptyset\)}
		\State Let label \(L = \left(\begin{array}{c} f(P_L)\\ Pred_L \end{array}\right)\) be a label in \(\mathcal{L}(v) \cap \texttt{Temp}\) with \(v\neq t\)
		\State \(\texttt{Temp} = \texttt{Temp}\backslash \{L\}\)
		\For{\((v,w) \in E\)}
			\State \(L' = L \oplus f(v,w)\)
			\If{\(L' \notin \mathcal{L}(w)\)}
				\If{\(w \neq t\)}
					\State \(\mathcal{L}(w) = \mathcal{L}(w) \cup \{L'\}\)
					\State \(\texttt{Temp} = \texttt{Temp}\cup \{L'\}\)
				\EndIf
				\If{\(w = t\)}
					\State \(\text{flag}(L') = 1\)
					\For{\(L^* \in \mathcal{L}(t)\)}
						\If{\(P_{L'} \precqq P_{L^*} \AND P_{L^*} \precnqq P_{L'}\)}
							\State \(\mathcal{L}(t) = \mathcal{L}(t)\backslash \{L^*\}\)
						\EndIf
						\If{\(P_{L^*} \precqq P_{L'} \AND P_{L'} \precnqq P_{L^*}\)}
							\State \(\text{flag}(L') = 0\)
						\EndIf
					\EndFor
					\If{\(\text{flag}(L') = 1\)}
						\State \(\mathcal{L}(t) = \mathcal{L}(t)\cup \{L'\}\)
					\EndIf
				\EndIf
			\EndIf	
		\EndFor
	\EndWhile
	\State \Return \(\mathcal{L}(t)\)
	\end{algorithmic}
\label{alg:ordinal}
\end{algorithm}

\begin{theorem}[Correctness]\mbox{}\label{thm:correctness}
	
	Algorithm \ref{alg:ordinal} correctly computes all ordinally non-dominated paths vectors from \(s\) to \(t\).
\end{theorem}
\begin{proof}
	To prove the correctness of the algorithm, we show that (a) after termination of the algorithm all ordinally non-dominated paths vectors from \(s\) to \(t\) defined by the labels in \(\mathcal{L}(t)\) are found and (b) that all labels in \(\mathcal{L}(t)\) define ordinally non-dominated paths vectors from \(s\) to \(t\).
	\begin{itemize}
		\item[(a)] For the first part assume that there exists a path \(P_{st}\) from \(s\) to \(t\) whose corresponding ordinal path vector \(f(P_{st})\) is non-dominated and not found by the algorithm. Further, suppose that its sorted version is not already contained in \(\mathcal{L}(t)\), i.e., \(\text{sort}(f(P_{st})) \neq \text{sort}(f(P_L))\) for all \(L \in \mathcal{L}(t)\). This case is only possible if the label \(\overline{L}\) corresponding to \(\text{sort}(f(P_{st}))\) has not been found. Consequently, we can assume that \(\overline{L}\) has never been included in \texttt{Temp}. Let \(u\) denote the predecessor node of \(t\) on \(P_{st}\) and let \(P_{su}\) denote the corresponding subpath of \(P_{st}\). Then, either the ordinal path vector \(f(P_{su})\) corresponding to subpath \(P_{su}\) has not been found or its sorted version was already contained in \(\mathcal{L}(u)\). If its sorted version was already contained in \(\mathcal{L}(u)\), then there exists a label \(L'\) at \(\mathcal{L}(u)\) with \(\text{sort}(f(P_{L'})) = \text{sort}(f(P_{su}))\). Extending \(L'\) by the arc \((u,t)\) would result in the same sorted ordinal path vector as for \(P_{st}\) which is a contradiction to our assumption. Thus, we may assume that \(P_{su}\) has not been found. Repeating this argument backwards along the path \(P_{st}\), we can see that for the first node  \(v \neq s\) on \(P_{st}\) the corresponding path, i.e., \(P_{sv}\), has not been found. This is impossible since we create the label \(L^* = (f(s,v), \lbrack sv\rbrack)^\top\) in the first iteration of the while-loop. 
		
		\item[(b)] The second part of the proof follows immediately from the algorithm (\texttt{lines 10 - 18}) as we only include a label in \(\mathcal{L}(t)\) if it is non-dominated. Further, for each non-dominated label which is found during the execution of the algorithm, we check whether it dominates non-dominated labels already contained in the label set at node \(t\). Thus, at termination of the algorithm all labels at \(\mathcal{L}(t)\) define ordinally non-dominated path vectors from \(s\) to \(t\), which concludes the proof.
	\end{itemize}
	
\end{proof}

\begin{theorem}[Running time]\mbox{}
	
	Algorithm \ref{alg:ordinal} has a worst-case running time complexity of \(\mathcal{O}(m\cdot \log n\cdot n^{3K+2})\).
\end{theorem}
\begin{proof}
	From Theorem \ref{thm:sizelist}, it follows that the number of ordinally non-dominated labels  at each node is bounded by \(\mathcal{O}(n^K)\). Since \(|V| = n\), the amount of work for the while-loop is bounded by \(\mathcal{O}(n^{K+1})\). In each iteration, we investigate all outgoing arcs from a node, create a new label \(L'\) and check if the created label is already contained in the label set at the end node of the arc under consideration. This can be done in \(\mathcal{O}(m\cdot n^{K})\). We conduct this operation at most \(n^{K+1}\) times which is bounded by \(\mathcal{O}(m\cdot n^{2K+1})\).
	If the end node of the considered arc is not the target node \(t\), we only execute constant time operations. If the end node is equal to \(t\), we have to do a dominance check for each label in the label set at node \(t\). For the dominance check we have to determine the lengths of the ordinal paths vectors which can be done in \(\mathcal{O}(n)\) and we have to sort both ordinal paths vectors which can be done in \(\mathcal{O}(n \log n)\) using a sorting algorithm, e.g., heapsort. Thus, in total, the dominance check for all labels in \(\mathcal{L}(t)\) can be done in \(\mathcal{O}(\log n\cdot n^{K+1})\).
	The remaining operations can be done in constant time. Consequently, the runtime is in \(\mathcal{O}(m\cdot \log n\cdot n^{3K+2})\), which concludes the proof.
	
\end{proof}

\begin{remark}\mbox{}
	
	Note, that \(\mathcal{O}(n^K)\) is a combinatorial asymptotic upper bound for the amount of different labels at node \(t\). However, many labels will be dominated after executing the dominance check such that the number of labels in \(\mathcal{L}(t)\) after termination of the algorithm can expected to be much smaller.
\end{remark}

Note that Bellman's principle of optimality does not hold. Thus, we are not able to delete labels at a node \(v\neq t\) until we reach node \(t\), cf. Example \ref{ex:alg+bellman}.
\begin{example}\label{ex:alg+bellman}\mbox{}
	
	The following example illustrates Algorithm \ref{alg:ordinal} and shows that Bellman's principle of optimality does not hold in general. After termination of the algorithm the label sets at each node are as depicted in Figure \ref{fig:alg+bellman}.
\begin{figure}[h!]
	\begin{minipage}{0.5\textwidth}
		\begin{itemize}
			\item[1)] \(L(s) = \left\{\left( \begin{array}{c} (\cdot)\\ \lbrack s\rbrack \end{array}\right)\right\}\)
			\item[2)] \(L(a) = \left\{\left( \begin{array}{c} (1)\\ \lbrack sa\rbrack \end{array}\right)\right\}\)
			\item[3)] \(L(b) = \left\{\left( \begin{array}{c} (1)\\ \lbrack sb\rbrack \end{array}\right), \left( \begin{array}{c} (1,2)\\ \lbrack sab\rbrack \end{array}\right) \right\}\)
			\item[4)] \(L(t) = \left\{\left( \begin{array}{c} (1,3)\\ \lbrack sbt\rbrack \end{array}\right), \left( \begin{array}{c} (1,2,3)\\ \lbrack sabt\rbrack \end{array}\right) \right\}\)
		\end{itemize}
	\end{minipage}
	\begin{minipage}{0.5\textwidth}
		\centering
		\begin{tikzpicture}[scale=0.5]
		\node[draw, circle,inner sep=2pt] (s) at (0,10) {$s$};
		\node[draw, circle,inner sep=2pt] (a) at (4,13) {$a$};
		\node[draw, circle,inner sep=2pt] (b) at (8,10) {$b$};
		\node[draw, circle,inner sep=2pt] (t) at (14,10) {$t$};
		
		\draw[->,thick] (s)--node[above]{\(1\)}(a);
		\draw [->,thick] (a)--node[above]{\(2\)}(b);
		\draw [->,thick] (s)--node[above]{\(1\)}(b);
		\draw [->,thick] (b)--node[above]{\(3\)}(t);
		
		\end{tikzpicture}
	\end{minipage}
\caption{Algorithm \ref{alg:ordinal} -- Example}
\label{fig:alg+bellman}
\end{figure}
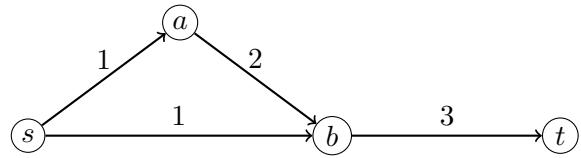
Initially, we create a label \(((\cdot),\lbrack s\rbrack)^\top\) at node \(s\) and put this label into \texttt{Temp}. During the first iteration, we remove this label from \texttt{Temp}, investigate all outgoing arcs from \(s\) and create new labels at nodes \(a\) and \(b\). Consequently, the list of temporary labels is as follows: \(\texttt{Temp} = \{((1),\lbrack sa\rbrack)^\top,((1),\lbrack sb\rbrack)^\top\}\). We choose the first label from \texttt{Temp} and investigate all outgoing arcs from node \(a\) and create a new label \(((1,2),\lbrack sab\rbrack)^\top\) at node \(b\) and put it into the set of temporary labels. If we repeat this procedure until \(\texttt{Temp} = \{((1,3),\lbrack sbt\rbrack)^\top, ((1,2,3),\lbrack sbt\rbrack)^\top\}\), one can verify that at node \(b\) the label \(((1,2),\lbrack sab\rbrack)^\top\) is dominated by \(((1),\lbrack sb\rbrack)^\top\). However, by deleting the dominated label we would not obtain all ordinally non-dominated paths vectors at node \(t\). Thus Bellman's principle of optimality does not hold in general. If we continue this procedure until \(\texttt{Temp} = \emptyset\), we get the label sets as depicted in Figure \ref{fig:alg+bellman}.
\end{example}

Further, the following example reveals the reason for considering directed acyclic graphs instead of general directed graphs.

\begin{example}\label{ex:acyclic}\mbox{}
	
	One can see that in Figure \ref{fig:acyclic} there are two different paths \(P_1 = (s,a,t)\) and \(P_2 = (s,c,b,a,t)\) corresponding to two different ordinally non-dominated paths vectors from \(s\) to \(t\), i.e., \(f(P_1) = (3,1)\) and \(f(P_2) = (3,3,1,1)\).
	
Since we are only interested in simple paths, one would replace \texttt{line 6} of Algorithm \ref{alg:ordinal} by
\begin{algorithm}[h!]
	\begin{algorithmic}
		\If{\(L' \notin \mathcal{L}(w) \AND w \notin Pred_{L}\)}
		\EndIf
	\end{algorithmic}
\end{algorithm}

in order to exclude cycling.

Then, however, if the label \(L = (3,3)\) at node \(b\) is created by the path \(P = (s,a,b)\) first, the label \(L' = (3,3)\) corresponding to \(P' = (s,c,b)\) will never be created at node \(b\). Thus, the ordinally non-dominated path vector \((3,3,1,1)\) will not be found by the algorithm.
	\begin{figure}[h!]
		\centering
		\begin{tikzpicture}
		\node[draw, circle,inner sep=2pt] (s) at (0,10) {$s$};
		\node[draw, circle,inner sep=2pt] (b) at (4,10) {$b$};
		\node[draw, circle,inner sep=2pt] (a) at (2,10) {$a$};
		\node[draw, circle,inner sep=2pt] (c) at (4,8) {$c$};
		\node[draw, circle,inner sep=2pt] (t) at (6,10) {$t$};

		\draw[->,thick] (s)--node[above]{\(3\)}(a);
		\draw [->,thick] (a)edge [bend right] node[above]{\(3\)}(b);
		\draw [->,thick] (b)--node[above]{\(3\)}(t);
		\draw [->,thick] (s)--node[above]{\(3\)}(c);
		\draw[->,thick] (c)--node[right]{$3$}(b);
		\draw[->,thick] (b)edge [bend right] node[above]{$1$}(a);
		\draw[->,thick] (a)edge [bend angle=60,bend left] node[above]{$1$}(t);
		\end{tikzpicture}
		\caption{Algorithm \ref{alg:ordinal} only works on acyclic graphs.}
		\label{fig:acyclic}
	\end{figure}
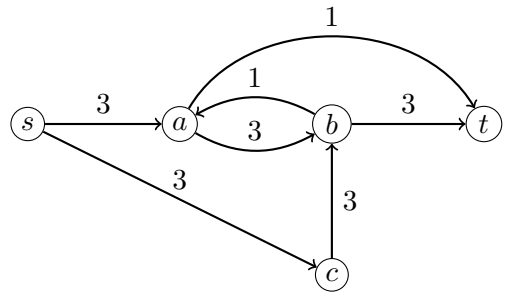
\end{example}

\begin{remark}\mbox{}
	
	By executing a brute-force algorithm, one could obtain the set of ordinally non-dominated paths vectors from \(s\) to \(t\) for general directed graphs: Just enumerate all simple paths from \(s\) to \(t\) and conduct a dominance check at node \(t\). Obviously, this approach yields an exponential worst-case complexity.
\end{remark}

\section{Practical Improvements}\label{sec:practicalImprovements}
In this section, we investigate two modifications of Algorithm \ref{alg:ordinal} for improving its efficiency in practical applications.

\subsection{Modification 1}
During the execution of Algorithm \ref{alg:ordinal}, new labels are created at each node \(v \in V\backslash\{t\}\) as long as this label is not already present at the node under consideration. We can improve this procedure by excluding a label \(L\) at node \(w \in V\) whenever there exists another label \(L'\) at the same node with \(\text{len}(f(P_{L})) = \text{len}(f(P_{L'}))\) and \(\text{sort}(f(P_{L'})) \preceq \text{sort}(f(P_L))\). 

Thus, we replace \texttt{lines 8 and 9} in Algorithm~\ref{alg:ordinal} with the modification described below, cf. Modification 1.
\begin{algorithm}[h!]
	\textbf{Modification 1}
	\vspace{-0.3cm}
	
	\noindent\rule{14.7cm}{0.4pt}
	\begin{algorithmic}
			\State \(\text{flag}(L') = 1\)
			\For{\(L^* \in \mathcal{L}(w)\)}
				\If{\(\text{len}(f(P_{L'})) = \text{len}(f(P_{L^*}))\)}
					\If{\(\text{sort}(f(P_{L'})) \preceq \text{sort}(f(P_{L^*}))\)}
						\State \(\mathcal{L}(w) = \mathcal{L}(w)\backslash\{L^*\}\)
					\EndIf
					\If{\(\text{sort}(f(P_{L^*})) \preceq \text{sort}(f(P_{L'}))\)}
						\State \(\text{flag}(L') = 0\)
					\EndIf
				\EndIf
			\EndFor
			\If{\(\text{flag}(L') = 1\)}
				\State \(\mathcal{L}(w) = \mathcal{L}(w) \cup \{L'\}\)
				\State \(\texttt{Temp} = \texttt{Temp} \cup \{L'\}\)
			\EndIf
	\end{algorithmic}
\end{algorithm}

\begin{corollary}[Correctness]\mbox{}
	
	Algorithm \ref{alg:ordinal} with Modification 1 correctly computes all ordinally non-dominated paths vectors from \(s\) to \(t\).
\end{corollary}
\begin{proof}
	It only remains to show that by replacing \texttt{lines 8 and 9} in Algorithm \ref{alg:ordinal} with the above mentioned modification we do not exclude any ordinally non-dominated paths vector from \(s\) to \(t\).
	Therefore, let \(v \in V\) be an arbitrary node in \(G\) and assume there exist two labels \(L_1\) and \(L_2\) with \(L_1, L_2 \in \mathcal{L}(v)\) satisfying \(\text{len}(f(P_{L_1})) = \text{len}(f(P_{L_2}))\) and \(\text{sort}(f(P_{L_1})) \preceq \text{sort}(f(P_{L_2}))\). Now, let \(P_{vt}\) be an arbitrary \(v\)-\(t\)-path and suppose both \(L_1\) and \(L_2\) extended by \(P_{vt}\) lead to ordinally non-dominated paths vectors from \(s\) to \(t\). If we extend \(L_1\) and \(L_2\) along the path \(P_{vt}\), we get \(\text{len}(f(P_{L_1}\cup P_{vt})) = \text{ len}(f(P_{L_2}\cup P_{vt}))\) and thus, \(\text{sort}(f(P_{L_1}\cup P_{vt})) \preceq \text{sort}(f(P_{L_2}\cup P_{vt}))\), which is a contradiction. Consequently, we do not exclude any ordinally non-dominated path vector from \(s\) to \(t\). The rest remains as in the proof of Theorem \ref{thm:correctness}.
	
\end{proof}

\begin{remark}\mbox{}
	
	Note that with Modification 1, we replace a constant time operation with a non-constant time operation. However, in practical applications, we reduce the amount of work by shrinking the number of labels in \(\mathcal{L}(v)\) for all \(v \in V\backslash\{t\}\) throughout the algorithm.
\end{remark}

%
%
%
%
%
%
%

\subsection{Modification 2}
Searching and removing a certain label \(L'\) in a set of labels \( \mathcal{L}(v) \) is difficult as long as the labels in \( \mathcal{L}(v) \) are not sorted. In the absence of a total order, labels cannot be stored sorted with respect to \(\precqq\). Nevertheless, total orders such as the lexicographic order exists on the domain of the label sets such that adding, searching, and removing a label can be implemented in time \(\mathcal{O}(K \log n)\). 

\begin{corollary}[Running Time]\mbox{}
	
	Algorithm \ref{alg:ordinal} with Modification 1 and 2 has a worst-case running time complexity of \(\mathcal{O}(m\cdot K\log^2n \cdot n^{2K+2})\).
\end{corollary}
\begin{remark}\mbox{}
	
	Note, that \(\mathcal{O}(K)\) can be assumed to be in \(\mathcal{O}(m)\) as the number of ordinal levels might be relatively small in practical applications, i.e., \(|\mathcal{C}| \ll |E|\).
\end{remark}

\section{Computational Results}\label{sec:computationalResults}
In this section, we investigate the performance of the proposed labeling algorithm with respect to both above mentioned practical improvements on randomly created acyclic graphs and on acyclic grid graphs. Note that Modification 2 uses both practical improvements.
We do not investigate the proposed labeling algorithm without any modification as it is a pure enumerative algorithm with inappropriate runtimes even for small graph instances.

All algorithms have been tested on a compute server equipped with two Intel Xeon E-5-2670 @ 2.6-3.3GHz and 96GB RAM, running on Ubuntu 12.04.5 LTS Server, Kernel 3.2.0 x86\_64.
For the computations, Python 2.7.14 and NetworkX 2.1 were used.

\subsection{Random Acyclic Graphs}
Random acyclic graph instances are created by specifying the number of nodes (25, 50, 100, 200), the number of ordinal levels (2, 3, 5, 8, 10) which are randomly assigned to each arc and the probability \(p\) that two nodes in the graph are connected (\(p \in \{0.2, 0.4, 0.6\}\)). For each instance with a fixed number of nodes we execute 150 instances.
However, for each distinct triple of "Number of Nodes", "Number of ordinal levels" and "\(p\)" we conduct 10 instances and show the minimal, mean and maximal runtimes of both modifications. Note that runtimes increase quite fast for randomly created acyclic graphs such that we restricted ourselves to 10 instances per triple.
Table 1 reveals that for an increasing number of nodes, an increasing number of ordinal levels as well for an increasing p the runtimes of both modifications increase. 
In Figure 6, one can see that for an arc creation probability of \(0.6\), Modification 2 outperforms Modification 1 as the number of nodes and the number of ordinal levels increase.
For \(K=2\), the runtimes of both modifications nearly coincide, see Table 1, such that we do not display both of them.
Note that we get similar results for \(p=0.4\) and \(p=0.2\).


\begin{center}
\scriptsize
	\label{tab:runtimeACGraphMod1}
	\begin{filecontents*}{ACGmod12f.csv}
x,y,z,min,w,mean,v,max,u
\(|V|\),\(|\mathcal{C}|\),p,Runtime Mod.1,Runtime Mod.2,Runtime Mod.1,Runtime Mod.2,Runtime Mod.1,Runtime Mod.2
25,2,0.2,0.1*,0.1*,0.1*,0.1*,0.1*,0.1*
,,0.4,0.1*,0.1*,0.1*,0.1*,0.1*,0.1*
,,0.6,0.1*,0.1*,0.1*,0.1*,0.1*,0.1*
25,3,0.2,0.1*,0.1*,0.1*,0.1*,0.1*,0.1*
,,0.4,0.1*,0.1*,0.1*,0.1*,0.1*,0.1*
,,0.6,0.1*,0.1*,0.1,0.1*,0.2,0.1
25,5,0.2,0.1*,0.1*,0.1*,0.1*,0.1*,0.1*
,,0.4,0.1*,0.1*,0.1*,0.1*,0.1*,0.1*
,,0.6,0.1*,0.1*,0.1,0.1,0.1,0.1
25,8,0.2,0.1*,0.1*,0.1*,0.1*,0.1*,0.1*
,,0.4,0.1*,0.1*,0.1*,0.1*,0.1*,0.1*
,,0.6,0.1*,0.1*,0.1,0.1,0.2,0.2
25,10,0.2,0.1*,0.1*,0.1*,0.1*,0.1*,0.1*
,,0.4,0.1*,0.1*,0.1*,0.1*,0.1*,0.1*
,,0.6,0.1*,0.1*,0.1*,0.1,0.2,0.2
50,2,0.2,0.1*,0.1*,0.1*,0.1*,0.1*,0.1*
,,0.4,0.2,0.2,0.2,0.2,0.3,0.3
,,0.6,0.5,0.5,0.6,0.6,0.7,0.7
50,3,0.2,0.1*,0.1*,0.1*,0.1*,0.1*,0.1*
,,0.4,0.1,0.1,0.2,0.3,0.4,0.4
,,0.6,0.4,0.5,0.8,0.8,1.3,1.3
50,5,0.2,0.1*,0.1*,0.1*,0.1*,0.1*,0.1*
,,0.4,0.2,0.2,0.4,0.4,0.7,0.7
,,0.6,0.7,0.7,1.4,1.5,1.9,1.9
50,8,0.2,0.1*,0.1*,0.1*,0.1*,0.2,0.2
,,0.4,0.5,0.5,0.9,0.9,1.5,1.5
,,0.6,1.6,1.7,3.0,3.1,4.6,4.4
50,10,0.2,0.1*,0.1*,0.1*,0.1*,0.2,0.2
,,0.4,0.4,0.4,0.9,0.9,1.4,1.4
,,0.6,1.3,1.4,3.5,3.5,7.0,6.5
100,2,0.2,0.3,0.3,0.5,0.5,0.7,0.7
,,0.4,1.8,1.9,2.5,2.6,3.2,3.3
,,0.6,5.7,5.9,6.3,6.5,7.0,7.1
100,3,0.2,0.5,0.6,0.8,0.8,1.2,1.1
,,0.4,2.6,2.8,4.4,4.4,8.6,8.4
,,0.6,9.8,10,12,12,14,14
100,5,0.2,0.7,0.8,1.1,1.2,1.7,1.7
,,0.4,5.3,5.6,12,12,22,21
,,0.6,18,18,42,39,79,70
100,8,0.2,1.0,1.1,2.0,2.1,3.8,3.7
,,0.4,13,13,35,31,55,48
,,0.6,35,34,129,110,251,204
100,10,0.2,1.1,1.2,2.1,2.2,3.6,3.6
,,0.4,24,23,39,35,70,61
,,0.6,156,136,221,174,304,226
200,2,0.2,3.9,4.0,6.5,6.5,8.4,8.3
,,0.4,28,28,34,34,38,37
,,0.6,59,59,75,73,91,87
200,3,0.2,5.7,5.9,11,11,13,13
,,0.4,51,51,82,79,126,118
,,0.6,146,146,255,235,365,323
200,5,0.2,27,27,55,51,124,108
,,0.4,384,294,612,413,775,501
,,0.6,537,497,1835,1413,3559,2492
200,8,0.2,90,82,109,96,144,119
,,0.4,1147,1149,2411,1642,3762,2135
,,0.6,8654,6453,18894,11359,45981,21122
200,10,0.2,68,64,173,138,338,233
,,0.4,1579,1668,4349,2676,7824,3724
,,0.6,6206,4573,27039,16585,49678,28273
\end{filecontents*}
\pgfplotstabletypeset[column type=,
	begin table={\begin{tabularx}{\textwidth}{p{0.5cm}| p{0.5cm}| p{0.5cm}|| X| X|| X| X|| X| X}},
		end table={\end{tabularx}},
	col sep=comma,
	string type,
	columns/x/.style={
		column name={},
		postproc cell content/.code={}
	},
	columns/y/.style={
		column name={},
		postproc cell content/.code={}
	},
	columns/z/.style={
		column name={},
		postproc cell content/.code={}
	},
	columns/w/.style={
		column name={},
		postproc cell content/.code={}
	},
	columns/v/.style={
		column name={},
		postproc cell content/.code={}
	},
	columns/u/.style={
		column name={},
		postproc cell content/.code={}
	},
	every nth row={3[+0]}{after row=\hline},
	every head row/.style={before row=\toprule, after row=\hline},
	every last row/.style={after row=\bottomrule},
	]{ACGmod12f.csv}
	\captionof*{table}{Table 1: Runtime of Modification 1 and 2 on Random Acyclic Graphs. The entry 0.1* indicates that the corresponding runtime of that instance is smaller than 0.1 sec.}
	\normalsize
\end{center}

\begin{figure}[h!]
	\centering
		\pgfplotstableread[col sep = comma]{plotDAG02.csv}\mydata
		\begin{tikzpicture}
		\begin{axis}[
		legend pos = north west,
		xlabel={Number of Nodes},
		ylabel={Mean Runtime [s]},
		xmin = 0,
		xmax = 200,
		ymin = 0,
		ymax = 28000
		]
		\addplot table[x = Number of Nodes, y = RT1]{\mydata};
		\addplot table[x = Number of Nodes, y = RT2]{\mydata};
		\addplot table[x = Number of Nodes, y = RT3]{\mydata};
		\scriptsize\legend{Runtime K=2,Runtime Mod.1 K=10, Runtime Mod.2 K=10}
		\end{axis}
		\end{tikzpicture}
	\captionof*{table}{Figure 6: Mean Runtime Modification 1 and 2 on Random Acyclic Graphs for \(p=0.6\) and \(K \in \{2,10\}\)}
\end{figure}

\subsection{Acylic Grid Graphs}
The nodes in the acyclic grid graph instances are arranged in a rectangular grid with given height and width. Every node has at most two outgoing arcs (up and right) to ensure that the graph is acyclic. Only nodes on the upper boundary have less outgoing arcs (the target node has no outgoing arc). As for the randomly created acyclic graph instances, the ordinal levels are randomly assigned to each arc and the runtimes of both modifications increase with an increasing number of nodes and an increasing number of ordinal levels. However, we are able to run the modifications on much bigger graph instances as the number of ordinal paths vectors with same length can expected to be very high due to the construction of the graph instances. Figure 7 indicates that Modification 2 outperforms Modification 1 although the difference is not as significant as for randomly created acyclic graphs.
Again, we note that the runtimes of both modification for \(K=2\) nearly coincide, such that we do not display both of them.

%
\scriptsize
\begin{center}
\label{tab:runtimeGridGraphMod1}
\begin{filecontents*}{GridMod12.csv}
x,y,min,z,mean,w,max,v
\(|V|\),\(|\mathcal{C}|\),Runtime Mod.1,Runtime Mod.2,Runtime Mod.1,Runtime Mod.2,Runtime Mod.1,Runtime Mod.2
200,2,0.1*,0.1*,0.1*,0.1*,0.1,0.1
,3,0.1*,0.1*,0.2,0.1,0.3,0.2
,5,0.1,0.1,0.3,0.2,0.6,0.5
,8,0.2,0.2,0.4,0.3,0.6,0.5
,10,0.2,0.2,0.4,0.4,0.7,0.5
500,2,0.2,0.2,0.3,0.2,0.5,0.3
,3,0.3,0.3,0.6,0.5,1.1,0.8
,5,0.8,0.6,1.6,1.4,2.9,2.1
,8,1.0,1.1,2.5,2.3,4.8,3.9
,10,1.3,1.3,2.9,2.8,6.0,6.0
1000,2,0.6,0.6,0.7,0.6,1.0,0.7
,3,1.0,1.0,2.0,1.9,3.6,2.9
,5,3.2,3.2,6.3,6.0,10,8.3
,8,6.3,6.3,11,11,15,15
,10,8.6,8.9,15,15,37,28
2000,2,1.5,1.4,2.0,1.7,3.4,2.1
,3,4.4,4.5,8.0,7.0,17,11
,5,15,15,31,30,58,49
,8,33,34,65,63,132,122
,10,44,45,98,90,171,159
3000,2,2.6,2.6,3.2,2.9,5.6,3.3
,3,8.1,8.2,16,15,30,23
,5,49,46,73,70,125,117
,8,81,84,168,159,294,281
,10,142,145,237,228,433,358
4000,2,4.0,4.1,4.1,4.2,4.2,4.2
,3,18,19,23,23,29,28
,5,69,70,111,111,137,134
,8,193,198,326,324,526,491
,10,309,315,473,464,687,632
5000,2,5.6,5.7,6.2,6.3,6.6,6.7
,3,25,26,39,37,55,48
,5,128,130,233,222,405,312
,8,379,384,655,644,1247,1115
,10,422,428,941,889,1560,1432
\end{filecontents*}
\pgfplotstabletypeset[column type=,
begin table={\begin{tabularx}{\textwidth}{p{1cm}| p{1cm}|| X| X|| X| X|| X| X}},
	end table={\end{tabularx}},
col sep=comma,
string type,
columns/x/.style={
	column name={},
	postproc cell content/.code={}
},
columns/y/.style={
	column name={},
	postproc cell content/.code={}
},
columns/z/.style={
	column name={},
	postproc cell content/.code={}
},
columns/w/.style={
	column name={},
	postproc cell content/.code={}
},
columns/v/.style={
	column name={},
	postproc cell content/.code={}
},
every nth row={5[+0]}{after row=\hline},
every head row/.style={before row=\toprule, after row=\hline},
every last row/.style={after row=\bottomrule},
]{GridMod12.csv}
\captionof*{table}{Table 2: Runtime of Modification 1 and 2 on Acyclic Grid Graphs. The entry 0.1* indicates that the corresponding runtime of that instance is smaller than 0.1 sec.}
\end{center}
\normalsize

\begin{figure}[h!]
	\centering
		\pgfplotstableread[col sep = comma]{plotDAG02r.csv}\mydata
		\begin{tikzpicture}
		\begin{axis}[
		legend pos = north west,
		xlabel={Number of Nodes},
		ylabel={Mean Runtime [s]},
		xmin = 0,
		xmax = 5000,
		ymin = 0,
		ymax = 1000
		]
		\addplot table[x = Number of Nodes, y = RT1]{\mydata};
		\addplot table[x = Number of Nodes, y = RT2]{\mydata};
		\addplot table[x = Number of Nodes, y = RT3]{\mydata};
		\scriptsize\legend{Runtime K=2, Runtime Mod.1 K=10, Runtime Mod.2 K=10}
		\end{axis}
		\end{tikzpicture}
		\captionof*{table}{Figure 7: Mean Runtime Modification 1 and 2 on Acyclic Grid Graphs for  \(K \in \{2,10\}\)}
\end{figure}	

\subsection{Comparison}
Figure 8 displays that even for a low arc creation probability, i.e., \(p=0.2\), the runtimes on random acyclic graphs exceed the runtimes on acyclic grid graphs intensively, cf. Figure 8(a).
In Figure 8(b) this gets even clearer as the difference between the runtime of Modification 2 on random acyclic graphs (16585 seconds) and the runtime on acyclic grid graphs (0.4 seconds) for \(K=10\) is roughly equal to 16584 seconds.
\begin{figure}[h!]
	\subfloat[Mean Runtime Modification 1 and 2 for \(N=200\), \(p=0.2\) and  \(K \in \{2,3,5,8,10\}\)]{
	\pgfplotstableread[col sep = comma]{plotGGvsDAG.csv}\mydata
	\begin{tikzpicture}[scale=0.76]
	\begin{axis}[
	legend pos = north west,
	xlabel={Number of ordinal levels},
	ylabel={Mean Runtime [s]},
	xmin = 2,
	xmax = 10,
	ymin = 0,
	ymax = 180
	]
	\addplot table[x = Number of ordinal levels, y = RT1]{\mydata};
	\addplot table[x = Number of ordinal levels, y = RT2]{\mydata};
	\addplot table[x = Number of ordinal levels, y = RT3]{\mydata};
	\scriptsize\legend{Runtime Mod.1 RAG, Runtime Mod.2 RAG, Runtime AGG}
	\end{axis}
	\end{tikzpicture}
}
\hspace{0.7cm}
\subfloat[Mean Runtime Modification 1 and 2 for \(N=200\), \(p = 0.6\) and  \(K \in \{2,3,5,8,10\}\)]{
	\pgfplotstableread[col sep = comma]{plotGGvsDAG1.csv}\mydata
	\begin{tikzpicture}[scale=0.76]
	\begin{axis}[
	legend pos = north west,
	xlabel={Number of ordinal levels},
	ylabel={Mean Runtime [s]},
	xmin = 2,
	xmax = 10,
	ymin = 0,
	ymax = 28000
	]
	\addplot table[x = Number of ordinal levels, y = RT1]{\mydata};
	\addplot table[x = Number of ordinal levels, y = RT2]{\mydata};
	\addplot table[x = Number of ordinal levels, y = RT3]{\mydata};
	\scriptsize\legend{Runtime Mod.1 RAG, Runtime Mod.2 RAG, Runtime AGG}
	\end{axis}
	\end{tikzpicture}
}
	\captionof*{table}{Figure 8: Mean Runtime Modification 1 and 2 on Random Acyclic Graphs (ACG) and Acyclic Grid Graphs (AGG) for \(N=200\), \(p\in \{0.2,0.6\}\) and  \(K \in \{2,3,5,8,10\}\)}
\end{figure}	

Consequently, in a real-world application one would prefer to model a specific task using acyclic grid graphs instead of general acyclic graphs.

\section{Conclusion}\label{sec:conclusion}
In this paper, we have introduced both the concepts of ordinal efficiency and ordinal non-dominance on acyclic graphs with ordinally weighted arc costs. We showed that the proposed order relation defined on the set of \(s\)-\(t\)-paths defines a preorder, but not a partial order. Further, we proved that the number of ordinally efficient paths might be exponential in the number of nodes, while the number of ordinally non-dominated paths vectors from \(s\) to every other node in the graph is polynomially bounded. Thus, we proposed a polynomial time labeling algorithm to compute the set of ordinally non-dominated paths vectors from \(s\) to \(t\). We showed correctness and runtime of the algorithm and proposed two practical improvements to reduce the running time in practical applications. We have conducted a computational study on two different types of test instances to illustrate the difference in efficiency of the proposed modifications. As expected, Modification 2 performed best for large graph instances and an increasing amount of ordinal levels.

In the future, one could possibly investigate other standard combinatorial optimization problems with ordinal values, e.g., the Knapsack problem with ordinally valued items or minimal spanning trees with ordinally weighted arcs. Further, one could investigate a multiobjective ordinal path problem by assigning two (or more) ordinal levels to every arc in the graph.
\newline

\textbf{\textsf{Acknowledgements}} This work was partially supported by the Bundesministerium f\"ur Bildung und Forschung (BMBF) under Grant No. 13N14561, Deutsche Forschungsgemeinschaft (DFG, Project-ID RU 1524/2-3) and the 
DAAD-CRUP Luso-German bilateral cooperation under the 2017-2018 research project MONO‐EMC (Multi‐Objective Network Optimization for Engineering and Management Support).
Jos\'{e} Rui Figueira also acknowledges the support from the FCT grant SFRH/BSAB/139892/2018.

\bibliographystyle{alpha}
\bibliography{Bibdesk}
\end{document}